\documentclass[11pt,a4paper,english]{amsart}
\usepackage{babel}
\usepackage[latin1]{inputenc}
\usepackage{amsmath}
\usepackage{amsthm}
\usepackage{amsfonts}
\usepackage{indentfirst}
\usepackage{graphicx}
\usepackage{amssymb}
\usepackage[mathscr]{eucal}
\usepackage{tikz}
\usepackage{caption}
\usepackage{amsthm}
\usepackage{amscd}
\newtheorem{teo}{Theorem}

\newtheorem{pro}[teo]{Proposition}
\newtheorem{lem}[teo]{Lemma}

\theoremstyle{definition}
\newtheorem{de}[teo]{Definition}
\newtheorem{exa}{Example}

\title{Quasi-Cyclic Constructions of Quantum Codes}
\author{Carlos Galindo, Fernando Hernando and Ryutaroh Matsumoto}
\curraddr{\texttt{Carlos Galindo and Fernando Hernando:} Instituto
Universitario de Matem\'aticas y Aplicaciones de Castell\'on and
Departamento de Matem\'aticas, Universitat Jaume I, Campus de Riu
Sec. 12071 Castell\'{o} (Spain)\\
\texttt{Ryutaroh Matsumoto:} Department of Information and
Communication Engineering, Nagoya University, Nagoya, 464-8603 Japan.}
\email{
galindo@uji.es;
carrillf@uji.es;
ryutaroh.matsumoto@nagoya-u.jp}
\date{}
\thanks{Supported by the Spanish Ministry of Economy/FEDER: grants MTM2015-65764-C3-2-P and MTM2015-69138-REDT, and the University Jaume I: grant PB1-1B2015-02}

\begin{document}

\begin{abstract}
We give sufficient conditions for self-orthogonality with respect to symplectic, Euclidean and Hermitian inner products of a wide family of quasi-cyclic codes of index two. We provide lower bounds for the symplectic weight and the minimum distance of the involved codes. Supported in the previous results, we show algebraic constructions of good quantum codes and determine their parameters.
\end{abstract}

\maketitle

\section*{Introduction}
Attention to quantum information processing, especially quantum computing, is rapidly growing,
as several companies seem to build quantum computers with many qubits \cite{natureq}.
One of the important theoretical techniques to realize quantum computation is the quantum error correction, which protects quantum memory and quantum computational process from noise.

Quantum error correction was proposed by Shor \cite{shor95}.
Its connection to classical error correction was mainly described in \cite{calderbank96,calderbank97,calderbank98,gottesman96,steane96}. Afterwards that connection was generalized to the non-binary case (see \cite{matsumotouematsu00,ashikhmin01, ashikhmin00}). Since then, the use of classical (error-correcting) codes has become one of standard methods for constructing quantum codes, see \cite{kkk} for a survey.

Quasi-cyclic codes (QC codes) are a generalization of classical cyclic codes. It is well-known that there are asymptotically good codes attaining the Gilbert-Varshamov bound among QC codes
\cite{kasami74,ling03}, so it is natural to use QC codes to construct good quantum codes. Hagiwara et al. \cite{hagiwara07,hagiwara11} studied constructions of quantum codes by QC LDPC codes.
They focused on long codes and probabilistic constructions.

In this paper, we consider a wide class of QC codes of index 2 (see Subsections \ref{the11} and \ref{theQC}) and give sufficient conditions for their self-orthogonality with respect to symplectic, Euclidean and Hermitian inner products. Sections \ref{qccq}, \ref{qce} and \ref{qch} are devoted to the symplectic, Euclidean and Hermitian cases, respectively. In addition, we get lower bounds for the symplectic weight and the minimum distance of the involved codes. As a consequence, we provide an algebraic construction of short stabilizer quantum codes coming from the previously introduced QC codes (see Theorems \ref{ssymp}, \ref{eeuc} and \ref{hher}). To testify the interest of our construction, we complete this paper by showing several examples of quantum codes with good parameters. Indeed, we get quantum codes exceeding the Gilbert-Varshamov bounds \cite{13Q, 17Q, matsumotouematsu00} and/or improving the parameters of those codes which could be obtained by the CSS procedure from the best known linear codes under the assumption of being self-orthogonal.

\section{Preliminaries}
In this section, we review the concept of quasi-cyclic  (QC) code and the existing connections between stabilizer quantum codes and classical codes. We also introduce the class of QC codes we will use.

Throughout the paper, $\mathbf{F}_q$ will denote the finite field with $q$ elements, $q$ being a positive power $p^r$ of a prime number $p$. Recall than an $[n,k,d]_q$ classical code is a linear space $C \subset \mathbf{F}_q^n$ of dimension $k$ and minimum (Hamming) distance $d$. For a set $S \subset \mathbf{F}_q^n$, $w(S)$ will denote the minimum of the Hamming weights of those vectors in $S$.

\subsection{Quasi-cyclic codes}
\label{the11}
For a vector $\vec{x} = (x_1, x_2, \ldots, x_{n}) \in \mathbf{F}_q^{n}$, we denote
\[
\sigma_1(\vec{x}) = (x_n, x_1, \ldots, x_{n-1}).
\]
A linear space $C \subset \mathbf{F}_q^{n}$ is said to be
{\it a cyclic  code} if $C = \sigma_1(C)$.

For a vector $\vec{x} \in \mathbf{F}_q^{2n}$, we denote
\[
\sigma_2(\vec{x}) = (x_n, x_1, \ldots, x_{n-1}, x_{2n}, x_{n+1}, \ldots, x_{2n-1}).
\]
A linear space $C \subset \mathbf{F}_q^{2n}$ is said to be
{\it a quasi-cyclic (QC) code} (of index 2) if $C = \sigma_2(C)$.

We denote by $(x^n-1)$  the ideal of the polynomial ring $\mathbf{F}_q[x]$ generated by
$x^n-1$, and by $R = \mathbf{F}_q[x]/(x^n-1)$  the quotient ring of $\mathbf{F}_q[x]$ modulo $(x^n-1)$.
Given a polynomial $g(x)$ in $\mathbf{F}_q[x]$, by $[g(x)]$ we mean its residue class in $R$. When studying cyclic codes, a vector $\vec{a} = (a_0$, \ldots, $a_{n-1})$ is identified with the residue class
\begin{equation}
[a(x)] = [a_0 + a_2 x + a_3 x^2 + \cdots + a_{n-1} x^{n-1}], \label{eq100}
\end{equation}
and $\sigma_1(\vec{a})$ corresponds to the class $[x a(x)]$. Thus,  a cyclic code can be identified with an ideal of $R$ via the correspondence (\ref{eq100}). Since $R$ is a principal ideal domain, any cyclic code  can be generated by a single $[g(x)]\in R$. In the sequel, the minimum Hamming distance of the cyclic code generated by $g(x)$ will be denoted by $d(g(x))$. The expression $g(x) | h(x)$, $g(x), h(x) \in \mathbf{F}_q[x]$ means that $g(x)$ divides $h(x)$. When $g(x) | h(x)$, the cyclic code generated by $[g(x)]$
contains that by $[h(x)]$.

A vector $\vec{c} = (a_0$, \ldots, $a_{n-1}$, $b_0$, \ldots, $b_{n-1})$ in $\mathbf{F}_q^{2n}$ can be identified with
$([a(x)]$, $[b(x)]) \in R^2$, where
\[
  a(x) = a_0 + a_2 x + a_3 x^2 + \cdots + a_{n-1} x^{n-1},
  b(x) = b_0 + b_2 x + b_3 x^2 + \cdots + b_{n-1} x^{n-1}.
\]
Then $\sigma_2(\vec{c})$ corresponds to the element
$([x a(x)], [x b(x)])$ in $R^2$.
By this correspondence,
we see that a QC code $C$ can be identified with an
$R$-submodule of $R^2$.

Note that a QC code generated by $m$ elements in $R^2$, $$([f_1(x)], [g_1(x)]), [f_2(x)], [g_2(x)]), \ldots,
([f_m(x)], [g_m(x)]),$$ can be regarded as the $R$-module
\[
\left\{ \sum_{i=1}^m ([a_i(x)f_i(x)], [a_i(x)g_i(x)]) \mid a_i(x) \in \mathbf{F}_q[x] \right\}.
\]

\subsection{Quantum code constructions from classical linear codes}
\label{qccc}

A stabilizer (quantum) code $\mathcal{C} \neq \{0\}$ is the common eigenspace of a commutative subgroup of the error group generated by a nice error basis on the space $\mathbf{C}^{q^n}$, where $\mathbf{C}$ denotes the complex numbers, $q$ is a positive power of a prime number and $n$ is a positive integer \cite{kkk}.  The code $\mathcal{C}$ has minimum distance $d$ as long as errors with weight less than $d$ can be detected or have no effect on $\mathcal{C}$ but some error with weight $d$ cannot be detected. Furthermore, if $\mathcal{C}$ has dimension $q^k$ as a $\mathbf{C}$-vector space, then we say that the  code $\mathcal{C}$ has parameters $[[n,k,d]]_q$.

For a linear space $C \subset \mathbf{F}_q^n$, $C^\perp$ denotes its Euclidean dual, that is $\{ \vec{x} \in \mathbf{F}_q^n \mid
\langle \vec{x}$, $\vec{y}\rangle = 0$, for all $\vec{y} \in C\}$, where $\langle \vec{x}$, $\vec{y}\rangle$ denotes the Euclidean (standard) inner product. From two classical linear codes $C_1$ and $C_2$ over $\mathbf{F}_q$ and assuming that $C_2 \subset C_1 \subset \mathbf{F}_q^n$,  we can construct a  stabilizer quantum code with parameters $$[[n, \dim C_1 - \dim C_2, \min\{ w(C_1 \setminus C_2), w(C_2^\perp, C_1^\perp)\}]]_q.$$  This construction was shown in \cite{ashikhmin00,calderbank96,steane96}.

Stabilizer quantum codes can also be constructed from classical self-orthogonal codes with respect to the Hermitian inner product (see for instance \cite[Corollary 16]{kkk}). Indeed, recall that the Hermitian inner product of two vectors $\vec{x} = (x_1, x_2, \ldots x_n)$ and $\vec{y} = (y_1, y_2, \ldots y_n)$ in $\mathbf{F}_{q^{2}}^n$ is defined as
\[
\langle \vec{x}, \vec{y} \rangle_h  := \sum_{i=1}^n x_i^q y_i.
\]
Now, if $C \subset \mathbf{F}_{q^2}^n$ is a classical code with parameters $[n,k,d]_{q^2}$ such that
\[
C^{\perp_h}:= \left\{ \vec{x} \in \mathbf{F}_{q^2}^n \mid
\langle \vec{x}, \vec{y}\rangle_h = 0 \right\} \subset C,
\]
then, it can be constructed a stabilizer quantum code with parameters $[[n,2k-n,d]]_q$.

Finally, we have another construction that can be seen in \cite{ashikhmin00}. For $\vec{x}$, $\vec{y} \in \mathbf{F}_q^{2n}$, their {\it symplectic inner product} is defined as
\[
\langle \vec{x}, \vec{y}\rangle_\mathrm{s}
= \sum_{i=1}^n x_i y_{n+i} - x_{n+i}y_i.
\]
Given a linear space $C \subset \mathbf{F}_q^{2n}$, we denote
\[
C^{\perp\mathrm{s}} = \{ \vec{x} \in \mathbf{F}_q^{2n} \mid \langle \vec{x}, \vec{y}\rangle_\mathrm{s}=0, \mbox{ for all } \vec{y} \in C \}.
\]
For $\vec{x} \in \mathbf{F}_q^{2n}$, set $w_s(\vec{x}) = \mathrm{card} \{ i \mid (x_i$, $x_{n+i}) \neq (0,0) \}$ and for a set $S \subset \mathbf{F}_q^{2n}$, we denote $w_s(S) = \min\{ w_s(\vec{x}) \mid
\vec{x} \in S \}$. We call $w_s$ as the {\it symplectic weight.}
 The result concerning stabilizer codes states that when $C \subset \mathbf{F}_q^{2n}$ is a linear code such that $C \supset C^{\perp\mathrm{s}}$, we can construct an $[[n, \dim C - n, w_s(C \setminus C^{\perp\mathrm{s}})]]_q$ stabilizer quantum code.

 \subsection{The supporting QC codes}
 \label{theQC}
 We devote this brief section to introduce the family of QC codes we are going to use for constructing stabilizer quantum codes. Recall that $p$ is a prime and $q=p^r$. Fix a positive integer $n$, consider the polynomial $x^n -1 \in \mathbf{F}_q[x]$ and assume that the splitting field of that
 polynomial is  $\mathbf{F}_{p^{mr}}$ for some positive integer $m$.

 Let $f(x), g(x)$ and $ h(x)$ be monic polynomials in $\mathbf{F}_q[x]$ whose degree is less than $n$ and such that both $f(x)$ and $g(x)$ divide $x^n -1$. Recall that the class $[f(x)]$ of a polynomial $f(x)$ as above that divides $x^n-1$ generates a cyclic code of length $n$ and dimension $n - \deg(f)$. Consider the check polynomial $f'(x)$ which satisfies $f(x) \cdot f'(x) = x^n -1$ and define
 \[
 f^\perp (x) : = x^{\deg f'} f'\left(\frac{1}{x}\right).
 \]
Then, it is well-known that $[f^\perp (x)]$ generates the dual code of the cyclic code generated by $[f(x)]$. Next we define the mentioned family of QC codes. We will use suitable subfamilies for obtaining our quantum codes.
\begin{de}
\label{theqc}
With the above notation, $Q_q(f,g,h)$ will be the QC code over $\mathbf{F}_q$ of length $2n$ generated by $([f(x)],[h(x)f(x)])$ and $(0,[g(x)])$. When $q$ and the polynomials be clear, we will denote it simply by $Q$.
\end{de}

Notice that, according to \cite[Section 2]{lally01}, the generator set of $Q_q(f,g,h)$ is a Groebner basis for the $\mathbf{F}_q[x]$-submodule $\psi^{-1}(Q_q(f,g,h))$, which is the preimage in $(\mathbf{F}_q[x])^2$ of the $R$-submodule $Q_q(f,g,h)$ under the class map $\psi: (\mathbf{F}_q[x])^2 \rightarrow R^2$.

To the best of our knowledge, this is the first family of QC codes of short length giving quantum codes by algebraic techniques. There was a first attempt in \cite{qian08} but it seems to be wrong because the proposed codes contradict the dimension formula for simple generator QC codes \cite{lally01}. Indeed, if one considers a QC code generated by a single polynomial vector
$([f_1(x)], [f_2(x)], \ldots, [f_\ell(x)])$, then $\dim C \leq n$ by \cite[Corollary 2.14]{lally01}. However $C \supset C^\perp$ implies $\dim C \geq \ell n / 2$, which is not satisfied by the codes in \cite{qian08}.

In our development and attached to polynomials $h(x) \in \mathbf{F}_q[x]$ with degree less than $n$, we will consider the polynomials $\bar{h}(x)$ defined as
\[
\bar{h} (x) : = x^{n} h\left(\frac{1}{x}\right).
 \]
They are instrumental as the following result shows.
\begin{pro}
\label{thebar}
Let $f(x), g(x)$ and $ h(x)$ be monic polynomials in $\mathbf{F}_q[x]$ whose degrees are less than $n$ and consider the vectors in $\mathbf{F}_q^n$ determined by their classes  in $R$ as described before Equality (\ref{eq100}). Then, the following equality of Euclidean inner products of vectors in $\mathbf{F}_q^n$ holds:
\begin{equation}
\langle [f(x)g(x)], [h(x)]\rangle = \langle [g(x)], [\overline{f}(x) h(x)] \rangle. \label{eq3}
\end{equation}
\end{pro}
\begin{proof}
  Let $f(x) = f_0 + \cdots + f_{n-1}x^{n-1}$. We have
  \[
  \langle [f(x)g(x)], [h(x)]\rangle =
  \sum_{i=0}^{n-1} f_i \langle [x^i g(x)], [h(x)] \rangle,
  \]
and
  \[
  \langle [g(x)], [\overline{f}(x)h(x)]\rangle =
  \sum_{i=0}^{n-1} f_i \langle [ g(x)], [x^{n-i}h(x)] \rangle.
  \]
  In order to prove the proposition, it is sufficient to show
  \begin{equation}
    \langle [x^i g(x)], [h(x)] \rangle = \langle [ g(x)], [x^{n-i}h(x)] \rangle.  \label{eq:r1}
  \end{equation}
  Let $g(x) = g_0 + \cdots + g_{n-1}x^{n-1}$ and
  $h(x) = h_0 + \cdots + h_{n-1}x^{n-1}$. Then
  \[
  \langle [x^i g(x)], [h(x)] \rangle = \sum_{j=0}^{n-1} g_{i+j \bmod n} h_j
  = \sum_{j=0}^{n-1} g_{j } h_{j+n-i \bmod n},\;\;\mathrm{and}
  \]
  \[
  \langle [ g(x)], [x^{n-i}h(x)] \rangle =  \sum_{j=0}^{n-1} g_{j } h_{j+n-i \bmod n},
  \]
  which shows Equality (\ref{eq:r1}).
\end{proof}

The following sections will explain how to get stabilizer quantum codes from suitable QC codes $Q_q(f,g,h)$ and will give information about their parameters.

\section{Quasi-cyclic construction of quantum codes with symplectic inner product}
\label{qccq}
\begin{pro}
\label{the5}
With the above notation, the QC code $Q:=Q_q(f,g,h)$ has dimension $2n - \deg(f(x)) - \deg(g(x))$, and
its symplectic dual $Q^{\perp\mathrm{s}}$ is quasi-cyclic and
generated by $([g^\perp(x)], [\overline{h}(x) g^\perp(x)])$
and $([0], [f^\perp(x)])$.
\end{pro}
\begin{proof}
The dimension can be deduced from the fact that $\{(f(x),h(x)f(x)),(0,g(x))\}$ is a Groebner basis for the preimage in $(\mathbf{F}_q[x])^2$ of the $R$-submodule $Q=Q_q(f,g,h)$ under the map $(\mathbf{F}_q[x])^2 \rightarrow R^2$ (see \cite[Theorem 2.1]{lally01}). Alternatively, it is clear that a codeword in $Q$ can be expressed as
$([a(x)f(x)], [a(x)h(x)f(x)+b(x)g(x)])$, where $\deg(a(x)) \leq n-\deg(f(x))-1$ and
$\deg(b(x)) \leq n-\deg(g(x))-1$. From $([a(x)f(x)], [a(x)h(x)f(x)+b(x)g(x)])$, one can determine unique representatives
$(a(x),b(x))$ with $\deg(a(x)) \leq n-\deg(f(x))-1$ and
$\deg(b(x)) \leq n-\deg(g(x))-1$, therefore the map $$([a(x)f(x)], [a(x)h(x)f(x)+b(x)g(x)]) \mapsto ([a(x)],[b(x)])$$ is an isomorphism of $\mathbf{F}_q$ linear spaces. Now the statement about dimension can be deduced from the fact that the cyclic code generated by $f(x)$ (respectively, $g(x)$) has dimension $n-\deg(f(x))$ (respectively,  $n-\deg(g(x))$).

With respect to duality, A vector generated by $([g^\perp(x)], [\overline{h}(x) g^\perp(x)])$
and $(0, [f^\perp(x)])$ has the form
$([c(x) g^\perp(x)], [c(x) \overline{h}(x) g^\perp(x)+d(x)f^\perp(x)])$,
whose symplectic inner product with
$([a(x)f(x)], [a(x)h(x)f(x)+b(x)g(x)])$ is, by Equality  (\ref{eq3}),
\begin{eqnarray*}
  &&\langle [a(x)f(x)], [d(x)f^\perp(x)]\rangle +
  \langle [a(x)f(x)], [c(x) \overline{h}(x) g^\perp(x)]\rangle \\
  && - \langle [a(x)h(x)f(x)], [c(x) g^\perp(x)]\rangle
  - \langle [b(x)g(x), c(x) g^\perp(x)]\rangle\\
  &=& \langle [a(x)f(x)], [c(x) \overline{h}(x) g^\perp(x)]\rangle
  - \langle [a(x)h(x)f(x), c(x) g^\perp(x)]\rangle\\
  &=& \langle [a(x)h(x)f(x)], [c(x) g^\perp(x)]\rangle -
  \langle [a(x)h(x)f(x)], [c(x) g^\perp(x)]\rangle\\
  &=& 0.
\end{eqnarray*}

This concludes the proof after taking into account that the dimension of
the space generated by $([g^\perp(x)], [\overline{h}(x) g^\perp(x)])$
and $([0], [f^\perp(x)])$ is $2n - \deg(g^\perp(x))- \deg(f^\perp(x))$.
\end{proof}

\begin{pro}
\label{the7}
Consider the  QC code $Q:=Q_q(f,g,h)$ where we assume that $h(x)$ satisfies that $\mathrm{gcd}(h(x)-\beta ,x^n-1) = 1$ for all non-zero $\beta\in\mathbf{F}_q$. Then, a lower bound on the symplectic weight of $Q$ is the following value
\begin{multline*}
  d_q(f,g,h) = \\
   \min \bigg\{ d([g(x)]), d\big([(x^n-1)/\mathrm{gcd}(x^n-1,h(x))]\big),
   d\Big(\big[\mathrm{lcm}\big(f(x),
g(x)/\mathrm{gcd}(g(x),h(x))\big)\big]\Big), \\
\Big(d([f(x)])+d([\mathrm{gcd}(h(x)f(x),g(x))])+(q-1)d([\mathrm{gcd}(f(x),g(x)))]\Big)\Big/q \bigg\}.
\end{multline*}
\end{pro}
\begin{proof}
Consider the symplectic weight $$w_s=w_s([a(x)f(x)], [a(x)h(x)f(x)+b(x)g(x)]).$$

If $[a(x)]=0$ then $w_s \geq d(g(x))$.

We are going to use the following relation among symplectic and Hamming weights of vectors $\{\vec{u}, \vec{v}\} \in \mathbf{F}_q^{2n}$ which was proved in \cite[Lemma 2.4]{ling10}.
\begin{equation}
  qw_s(\vec{u}, \vec{v}) = w_H(\vec{u})+w_H(\vec{v}) + \sum_{0\neq \alpha \in \mathbf{F}_q}
  w_H(\alpha\vec{u}+\vec{v}). \label{eq2}
\end{equation}

Suppose that $[b(x)]=0$, $[a(x)]\neq 0$ and $[a(x)h(x)f(x)] \neq 0$. Since $h(x)-\beta$ is a unit modulo $x^n-1$, $[a(x)(h(x)-\beta)f(x)]\neq 0$
for nonzero $\beta$. Then, for $q=2$, it holds
\begin{eqnarray*}
  w_s &=& \Big(w_H([a(x)f(x)]) + w_H ([a(x)h(x)f(x)]) +
  w_H ([a(x)(h(x)+1)f(x)])\Big)\Big/2 \\
  & \geq & \Big(d([f(x)]) + d([h(x)f(x)]) + d([f(x)])\Big)\Big/2\\
  & \geq & \Big(d([f(x)])+d([\gcd(h(x)f(x),g(x))])
  +d([\mathrm{gcd}(f(x),g(x))])\Big)\Big/2.
\end{eqnarray*}
For $q>2$, we have
\begin{eqnarray*}
  w_s &=& \Big(w_H([a(x)f(x)]) + w_H ([a(x)h(x)f(x)]) \\ && + \sum_{0\neq\beta\in\mathbf{F}_q}w_H ([a(x)(h(x)+\beta)f(x)])\Big)\Big/q \\
  & \geq & \Big(d([f(x)]) + d([\mathrm{gcd}(h(x)f(x), g(x))]) + (q-1)d([\mathrm{gcd}(f(x),g(x))])\Big)\Big/q.
\end{eqnarray*}

Suppose now that $[b(x)]=0$, $[a(x)]\neq 0$ and $[a(x)h(x)f(x)] = 0$. Then $w_s$ equals $w_H([a(x)f(x)])$ and $[a(x)f(x)]$ belongs to the
cyclic code generated by $[(x^n-1)/\mathrm{gcd}(x^n-1,h(x))]$. Thus $w_s \geq d([(x^n-1)/\mathrm{gcd}(x^n-1,h(x))])$.

Finally and until the end of the proof, we assume $[a(x)]\neq 0$ and $[b(x)] \neq 0$. Then, we have
\begin{multline}
\label{eq1}
qw_s = w_H([a(x)f(x)]) + w_H([a(x)h(x)f(x)+b(x)g(x)])\\
+\sum_{0\neq \beta \in \mathbf{F}_q}w_H([a(x)(h(x)+\beta)f(x)+b(x)g(x)]).
\end{multline}
If some summand of the summation in (\ref{eq1}) is zero,
then $[a(x)(h(x)+\beta)f(x)] = -[b(x)g(x)]$ for some $\beta \in \mathbf{F}_q$, which
means that $\mathrm{lcm}(f(x),g(x))|a(x)f(x)$ as $h(x)+\beta$ is a unit. So $$w_s \geq w_H([a(x)f(x)]) \geq d\big(\mathrm{lcm}(f(x),g(x))\big).$$
In case the second summand in (\ref{eq1}) is zero,
we get $$w_s=w_H([a(x)f(x)])$$ and $[a(x)f(x)]$ belongs to
the cyclic code generated by $[g(x)/\mathrm{gcd}(g(x),h(x))]$.
So $$w_s \geq d\Big(\big[\mathrm{lcm}\big(f(x), g(x)/\mathrm{gcd}(g(x),h(x))\big)\big]\Big).$$
Otherwise (all summands in (\ref{eq1}) are nonzero),
$$w_s \geq \Big(d([f(x)]) + d([\mathrm{gcd}(h(x)f(x), g(x))]) + (q-1)d([\mathrm{gcd}(f(x),g(x))])\Big)\Big/q,$$ which concludes the proof.
\end{proof}

\begin{teo}
\label{ssymp}
With the above notation, assume that the polynomial $h(x)$ satisfies that $\gcd(h(x)-\beta ,x^n-1) = 1$ for all non-zero $\beta\in\mathbf{F}_q$. Assume also that it holds either (i) $f(x)|g^\perp(x)$, $g(x)|f^\perp(x)$ and $h(x)|\overline{h}(x)$, or (ii) $f(x)|g(x)|g^\perp(x)|f^\perp(x)$.

Then,  the  QC code $Q:=Q_q(f,g,h)$ is symplectic self-orthogonal and allows us to construct a stabilizer quantum code with parameters $[[n, n - \deg(f(x)) - \deg(g(x)), \geq d_q(f,g,h)]]_q$.
\end{teo}

\begin{proof}
The fact that $Q$ is self-orthogonal follows trivially from Proposition \ref{the5} in Case (i). In Case (ii), we have $f^\perp(x) = \alpha_1(x) g^\perp(x) $, $g^\perp(x) = \alpha_2(x) g(x)$ and $g(x) = \alpha_3(x) f(x)$, where $\alpha_i(x) \in \mathbf{F}_q[x]$ for $i \leq i \leq 3$. Now
\begin{multline*}
\Big(a(x)g^\perp(x), a(x) \bar{h}(x) g^\perp(x) + b(x) f^\perp (x) \Big)\\
= \Big( a(x) \alpha_2(x) \alpha_3(x) f(x), a(x) \alpha_2(x) \alpha_3(x) f(x) + q(x) g(x) \Big),
\end{multline*}
where $q(x)= a(x) \alpha_2(x) (\bar{h}(x) -1) + b(x) \alpha_1(x) \alpha_2(x)$, which again by Proposition \ref{the5}, proves the self-orthogonality in this case. Now Proposition \ref{the7} and Subsection \ref{qccc} conclude the proof.
\end{proof}

To finish this section, we will provide some polynomials $h(x)$ which are suitable for the previous mentioned purposes.

For each set $\{i, j\}$ of positive integers, consider the following trace polynomials in $\mathbf{F}_q[x]$
\[
\mathrm{tr}_{ji/i}(x) = x+ x^{p^i} + x^{p^{2i}} + \cdots + x^{p^{(j-1)i}}.
\]

\begin{pro}
\label{trace}
Assume, as above, that the splitting field of $x^n -1 \in \mathbf{F}_q[x]$ is  $\mathbf{F}_{p^{mr}}$ and consider a positive integer $s<p$ which divides $m$ and is coprime with $r$. Then the polynomial in $\mathbf{F}_q[x]$
\[
h(x) = (p-s) \mathrm{tr}_{mr/s}(x) + \mathrm{tr}_{mr/1}(x)
\]
satisfies that $h(x) + \beta$ is coprime with $x^n-1$ for all $\beta \in
\mathbf{F}_q\setminus \{0\}$.
\end{pro}
\begin{proof}
In this proof, for the sake of simplicity, we will use the same expression for the involved polynomials and the maps which they define. We are going to prove that the equation $h(x)+\beta = 0$ has no solution in $\mathbf{F}_{q^m}$,
which is equivalent to $h(a)+\beta \neq 0$ for all $a\in \mathbf{F}_{q^m}$
and $\beta \in \mathbf{F}_{q}\setminus \{0\}$.

Indeed, observe that $\mathrm{tr}_{j i/i}$ can be regarded as a map $\mathbf{F}_{p^{ji}} \rightarrow \mathbf{F}_{p^{i}} $. In addition, it holds the equality $\mathrm{tr}_{m  r/1} = \mathrm{tr}_{s/1} \circ \mathrm{tr}_{m  r/s}$, where $\circ$ means maps composition. When $\mathrm{tr}_{m  r/s}(a)=b \in \mathbf{F}_p$, we have $\mathrm{tr}_{m  r/1}(a) = s b$ and $h(a) = p b= 0$. Otherwise we have $\mathrm{tr}_{m  r/s}(a)=b' \in \mathbf{F}_{p^s}\setminus \mathbf{F}_p$,
which cannot be equal to $-\mathrm{tr}_{m  r/1}(a) - \beta \in \mathbf{F}_q$,
because our conditions imply that $(\mathbf{F}_{p^s}\setminus \mathbf{F}_p) \cap \mathbf{F}_q
=\emptyset$.
\end{proof}


Polynomials $h(x)$ in Proposition \ref{trace} need not to be of degree less than $n$ but this condition can be obtained by considering the remainder $h'(x)$ of $h(x)$ by division on $x^n-1$. The fact that $h'(x)$ satisfies the conclusion of Proposition \ref{trace} can be easily proved from B\'ezout's identity.

Finally we explain when the polynomials $h(x)=x+1$ or $h(x)=x^p-x$ are suitable for our purposes.

\begin{lem}
\label{linear}
With the above notation, it holds that $\gcd (x^n-1,x+1+\beta)= 1$ for all $\beta \in \mathbf{F}_{q} \setminus \{0\}$  if and only if $\gcd (q-1,n)=1$.
\end{lem}
\begin{proof}
Assume  that $\gcd (x^n-1,x+1+\beta)= x+1+\beta = x - \alpha$, $\alpha \in \mathbf{F}_{q}$, which means that  $x^n-1$ contains a $q-1$ root of unity and so
$\alpha^{q -1 } = \alpha^n = 1$. This equality holds if and only if $\gcd (q-1,n) \neq 1$, which concludes the proof.
\end{proof}

\begin{pro}
\label{theh}
The polynomial $h(x)=x+1 \in \mathbf{F}_{2}[x]$ ($h(x)=x^p-x \in \mathbf{F}_{p}[x]$, respectively) satisfies  the condition $\gcd(h(x)-\beta , x^n-1) = 1$ for all non-zero $\beta\in\mathbf{F}_q$, being $q=2$ (respectively, $q=p$ and $p$ does not divide $m=\log_p (n+1)$).
\end{pro}
\begin{proof}
Lemma \ref{linear} proves the case $q=2$. When $q=p$ is a prime number which does not divide $m=\log_p (n+1)$, we use the fact that $x^p-x + \beta$ is irreducible over $\mathbf{F}_{p^m}$ if and only if $\mathrm{tr}_{\mathbf{F}_{p^m}/\mathbf{F}_p}(\beta) \neq 0$ \cite[Corollary 3.79]{niederreiter08}. Then  $x^p-x + \beta$ is irreducible over $\mathbf{F}_{p^m}$ and therefore $x^p-x + \beta$ is coprime with $x^n-1$ for $n=p^m-1$.
\end{proof}

\section{Quasi-cyclic construction of stabilizer quantum codes with the Euclidean
inner product}
\label{qce}
Let $Q_q(f,g,h)$  be the QC code over $\mathbf{F}_q$ of length $2n$ generated by $([f(x)],[h(x)f(x)])$ and $(0,[g(x)])$ as introduced in Subsection \ref{theqc}. We are going to study the stabilizer quantum codes given by self-orthogonal codes with respect to Euclidean inner product of the form $Q_q(f,g,h)$. This way of obtaining quantum codes is usually known as the CSS construction \cite{calderbank96, steane95}.  For a start, we explain which code is the Euclidean dual of $Q_q(f,g,h)$.

\begin{pro}
\label{theEu}
The Euclidean dual code of the QC code $Q_q(f,g,h)$ over $\mathbf{F}_q$  is a QC code generated by the pairs $([-\overline{h}(x)g^\perp(x)], [g^\perp(x)])$ and $([f^\perp(x)], 0)$.
\end{pro}
\begin{proof}
  A codeword in $Q_q(f,g,h)$
  can be written as
  $([a_1(x)f(x)]$,
  $[a_1(x)h(x)f(x)+a_2(x)g(x)])$.
  Similarily, a codeword in the code generated by
  $([-\overline{h}(x)g^\perp(x)], [g^\perp(x)])$ and $([f^\perp(x)], 0)$
  can be written as
  $([-b_1(x)\overline{h}(x)g^\perp(x)+b_2(x)f^\perp(x)]$,
  $[b_1(x)g^\perp(x)])$.
  The Euclidean inner product of the above two codewords is,
  by Proposition \ref{thebar},
  \begin{eqnarray*}
    && - \langle [a_1(x)f(x)], b_1(x)\overline{h}(x)g^\perp(x) \rangle
    + \underbrace{\langle [a_1(x)f(x)],  [b_2(x)f^\perp(x)]\rangle}_{=0}\\&&
    + \langle [a_1(x)h(x)f(x)], [b_1(x)g^\perp(x) ]\rangle
    + \underbrace{\langle [a_2(x)g(x)], [b_1(x)g^\perp(x) ]\rangle}_{=0}\\
    &=& - \langle [a_1(x)f(x)], b_1(x)\overline{h}(x)g^\perp(x) \rangle
    + \langle [a_1(x)h(x)f(x)], [b_1(x)g^\perp(x) ]\rangle\\
    &=& - \langle [a_1(x)f(x)], b_1(x)\overline{h}(x)g^\perp(x) \rangle
    + \langle [a_1(x)f(x)], [b_1(x)\overline{h}(x)g^\perp(x) ]\rangle = 0.
  \end{eqnarray*}
  We have shown that
  the Euclidean dual code of $Q_q(f,g,h)$
  contains the QC code generated by $([-\overline{h}(x)g^\perp(x)], [g^\perp(x)])$ and $([f^\perp(x)], 0)$. As is Proposition \ref{the5},
  the dimension of $Q_q(f,g,h)$ is
  $2n - \deg f(x) - \deg g(x)$, and
  that of the latter is
  $2n - \deg f^\perp(x) - \deg g^\perp(x) = \deg f(x)+\deg g(x)$,
  which completes the proof.

\end{proof}

Now, we give conditions for self-orthogonality.

\begin{pro}
\label{theEuself}
A sufficient condition for $Q_q(f,g,h)$ to contain its Euclidean dual is $$f(x) | g(x) | g^\perp(x) | f^\perp(x).$$
\end{pro}
\begin{proof}
It follows from the following two equalities:
\begin{multline*}
   \left([\overline{h}(x)g^\perp(x)], [g^\perp(x)]\right) - \left[\frac{\overline{h}(x)g^\perp(x)}{f(x)}\right]([f(x)], [h(x)f(x)])\\
  = \left(0, \left[\left(1-h(x)\overline{h}(x)\right) g^\perp(x)\right]\right) \\
  = \left[\frac{(1-h(x)\overline{h}(x))g^\perp(x)}{g(x)}\right] \left(0,[g(x)]\right).
\end{multline*}

\begin{multline*}
   \left([f^\perp(x)], 0 \right)\\
  = \left[ \frac{f^\perp(x)}{f(x)} \right] \left([f(x)],[h(x)f(x)]\right) -
  \left[ \frac{[h(x)f^\perp(x)]}{[g(x)]}\right]\left(0,[g(x)]\right).
\end{multline*}
\end{proof}

With respect to distance, we can state the following result.
\begin{pro}
\label{theEudist}
The following value
\begin{multline*}
  d_q^e(f,g,h) = \\
   \min \bigg\{ d([g(x)]), d\big([(x^n-1)/\gcd(x^n-1,h(x))]\big),
   d\Big(\big[\mathrm{lcm}\big(f(x),
g(x)/\gcd(g(x),h(x))\big)\big]\Big), \\
d([f(x)])+d([\gcd(h(x)f(x),g(x))]) \bigg\}
\end{multline*}
is a lower bound for the minimum distance of the  QC code $Q_q(f,g,h)$.
\end{pro}
\begin{proof}
A codeword in $Q_q(f,g,h)$ can be written as
$\big([a(x)f(x)], [a(x)h(x)f(x)+b(x)g(x)]\big)$.

If $[b(x)]=0$ and $[a(x)h(x)f(x)]\neq 0$, then its Hamming weight is at least $d(f(x))+d(h(x)f(x))$.

If $[b(x)]=0$ and $[a(x)h(x)f(x)]=0$, then then $a(x)f(x)$ belongs to the ideal in $\mathbf{F}_q[x]$ generated by $(x^n-1)/\mathrm{gcd}(x^n-1,h(x))$. So $w_H([a(x)h(x)]) \geq d((x^n-1)/\mathrm{gcd}(x^n-1,h(x)))$.

If $[a(x)]=0$ then the Hamming weight is larger than or equal to $d(g(x))$.

Finally and until the end of proof, we assume $[a(x)]\neq 0 $ and $[b(x)] \neq 0$. If $[a(x)h(x)f(x)+b(x)g(x)] = 0$ then $[a(x)f(x)]$ belongs to the cyclic code generated
by $[g(x)/\mathrm{gcd}(g(x),h(x))]$. Indeed, set $m(x)= \gcd(g(x),h(x))$ and, as a consequence, $g(x)=m(x) g'(x)$, $h(x)=m(x) h'(x)$ and $\gcd(g'(x),h'(x))=1$. Then $[a(x)h(x)f(x)+b(x)g(x)] = 0$ implies $$a(x)h(x)f(x) \in (g(x))$$ because $g(x) | x^n -1$. Thus $a(x)m(x)h'(x)f(x)=p(x)m(x)g'(x)$ for some polynomial $p(x)$ which proves that $a(x)f(x)$ belongs to the ideal generated by $g'(x)$. So
$$w_H\left([a(x)f(x)]\right) \geq d\left(\mathrm{lcm}(f(x),g(x)/\mathrm{gcd}(g(x),h(x)))\right).$$
Otherwise ($[a(x)h(x)f(x)+b(x)g(x)] \neq 0$) and
then
\begin{multline*}
w_H\left(([a(x)f(x)], [a(x)h(x)f(x)+b(x)g(x)])\right) \\
\geq d\left(f(x)) + d(\gcd(h(x)f(x), g(x))\right),
\end{multline*}
which concludes the proof.
\end{proof}

\begin{teo}
\label{eeuc}
With the above notation, assume that the polynomials $f(x)$ and $g(x)$ satisfy that $f(x) | g(x) | g^\perp(x) | f^\perp(x)$, then the QC code $Q_q(f,g,h)$  is self-orthogonal for the Euclidean inner product and, as a consequence, it provides a stabilizer quantum code with parameters $$[[2n,  2n - 2\deg(f(x)) - 2\deg(g(x)), \geq d_q^e(f,g,h)]]_q.$$
\end{teo}
\begin{proof}
It follows from Propositions \ref{theEu}, \ref{theEuself} and \ref{theEudist}, and Subsection \ref{qccc}.
\end{proof}

\section{Quasi-cyclic construction of quantum codes with the Hermitian
inner product}
\label{qch}
In this section the coefficient field for our QC codes and polynomials will be $\mathbf{F}_{q^2}$. This fact will allow us to consider Hermitian inner product instead of Euclidean inner product. Recall that for two vectors $\vec{x}$, $\vec{y} \in \mathbf{F}_{q^2}^{2n}$, the Hermitian inner product $\langle \vec{x},
\vec{y} \rangle_h$ can be regarded as the Euclidean product $\langle \vec{x}^q, \vec{y} \rangle$, where
$\vec{x}^q$ denotes component-wise $q$th power of the vector $\vec{x}$.

Denote by $Q_{q^2}(f,g,h)$ the QC code in $\mathbf{F}_{q^2}^{2n}$
generated by $([f(x)], [h(x)f(x)])$ and $(0,[g(x)])$. Attached to a polynomial $r(x) = a_0 + a_1 x + \cdots + a_m x^m$ of degree $m <n$, we define
 $r^{[q]}(x) = a_0^q + a_1^q x + \cdots + a_m^q x^m$. If $\vec{x}$ is represented by $[f(x)]$ then
$\vec{x}^q$ is represented by $[f^{[q]}(x)]$.


\begin{pro}
\label{theHe}
The Hermitian dual code of the QC code over $\mathbf{F}_{q^2}$ $Q_{q^2}(f,g,h)$ is a QC code generated by the pairs $([-\overline{h^{[q]}}(x){g^{[q]}}^\perp(x)],
[{g^{[q]}}^\perp(x)])$ and $([{f^{[q]}}^\perp(x)], 0)$.
\end{pro}
\begin{proof}
The dimension of the Hermitian dual code of $Q_{q^2}(f,g,h)$ is $$\deg(f(x))+\deg(g(x)) = 2n - \deg {f^{[q]}}^\perp(x)
- \deg {g^{[q]}}^\perp(x).$$

Therefore, it suffices to check the following chain of equalities:
\begin{multline*}
\left\langle \left([f(x)], [h(x)f(x)]\right),\left(\left[-\overline{h^{[q]}}(x){g^{[q]}}^\perp(x)\right], [{g^{[q]}}^\perp(x)]\right)\right\rangle_h \\
= -\left\langle [f(x)], \left[\overline{h^{[q]}}(x){g^{[q]}}^\perp(x)\right] \right\rangle +
  \left\langle [h(x)f(x)], \left[{g^{[q]}}^\perp(x)\right] \right\rangle_h\\
  = -\left\langle [h(x)f(x)], \left[{g^{[q]}}^\perp(x)\right] \right\rangle + \left\langle [h(x)f(x)], \left[{g^{[q]}}^\perp(x)\right] \right\rangle_h\\
  =0.
\end{multline*}
\end{proof}

The following result can be proved with a similar reasoning as in Proposition \ref{theEuself}.
\begin{pro}
\label{theHeself}
A sufficient condition for $Q_{q^2}(f,g,h)$ to contain its Hermitian dual is $$f(x) | g(x) | {g^{[q]}}^\perp(x) | {f^{[q]}}^\perp(x).$$
\end{pro}

As a consequence of the previous results, we obtain the following one involving stabilizer quantum codes.

\begin{teo}
\label{hher}
With the above notation, assume that the above polynomials, with coefficients in $\mathbf{F}_{q^2}$, $f(x)$ and $g(x)$ satisfy that $f(x) | g(x) | {g^{[q]}}^\perp(x) | {f^{[q]}}^\perp(x)$, then the QC code $Q_{q^2}(f,g,h)$  is self-orthogonal for the Hermitian inner product and, as a consequence, it provides an stabilizer quantum code with parameters $$[[2n,  2n - 2\deg(f(x)) - 2\deg(g(x)), \geq d_{q^2}^e(f,g,h)]]_q.$$
\end{teo}
\begin{proof}
  It follows from what we said in Subsection \ref{qccc} with respect to Hermitian duality and the fact that $d_{q^2}^e(f,g,h)$ is a lower bound for the minimum distance of the QC code $Q_{q^2}(f,g,h)$.
\end{proof}

\section{Examples}
We devote this section to provide some examples of good stabilizer quantum codes coming from our constructions.

The two first examples use symplectic product as explained in Section \ref{qccq}.
\begin{exa}
\label{ex1}
{\rm
Set $n=151$, $q=2$ and the polynomial $x^{151}-1 \in \mathbf{F}_2[x]$. The splitting field of $x^{151}-1$ is $\mathbf{F}_{2^{15}}$ and set $\zeta$ a primitive element. Taking cyclotomic cosets modulo $n=151$ with respect to $q=2$, one can get minimal polynomials of roots  of $x^{151}-1$ which divide that polynomial.

Consider the cyclotomic coset
$$\{  2, 4, 8, 16, 32, 64, 128, 105, 59, 118, 85, 19, 38, 76, 1  \},$$
and the attached polynomial $f(x)=(x-\zeta^2)(x-\zeta^4) \cdots (x-\zeta)$ which belongs to $\mathbf{F}_2[x]$ and divides $x^{151}-1$. Analogously, let $g(x) \in \mathbf{F}_2[x]$ be the polynomial defined by the next two cyclotomic cosets:
$$\big\{ [ 2, 4, 8, 16, 32, 64, 128, 105, 59, 118, 85, 19, 38, 76, 1 ], \mbox{ and}$$ $$ [ 10, 20, 40, 80, 9, 18, 36, 72, 144, 137, 123, 95, 39, 78, 5 ] \big\}.$$
If now $h(x)=x+1$, using the QC code $Q_{2}(f,g,h)$ and Theorem \ref{ssymp} we are able to construct a {\it stabilizer quantum code with parameters} $[[151,106,8]]_2$.

To testify the goodness of this code, we note that  \cite{edel} only gives a quantum code with parameters $[[151,106,6]]_2$. In addition, according to \cite{grassl}, a code with parameters $[151,128,8]_2$ is the best known binary linear code with length $151$ and minimum distance $8$. In the unlikely case it were self-orthogonal for the Euclidean inner product, by using the CSS construction, we would get a $[[151,105,8]]_2$ code, with one unit less of dimension than our code.

}
\end{exa}

\begin{exa}
\label{ex2}
{\rm In this example we use again Theorem \ref{ssymp} for providing a stabilizer quantum code with good parameters.

Set $n=73$ and $q=2^3$. The splitting field of  $x^{73}-1$ is $\mathbf{F}_{2^{9}}$. As in Example \ref{ex1}, considering a primitive element of this field, we consider the polynomial  $f(x)$ (respectively, $g(x)$) in $\mathbf{F}_8[x]$ defined by the cyclotomic coset $$\{  8, 64, 1 , 16, 55, 2 ,  24, 46, 3 \}$$ (respectively, $$\{  8, 64, 1 , 16, 55, 2 ,  24, 46, 3 ,   56, 10, 7 \}.)$$
Taking $h(x)=x+1$, from $Q_2(f,g,h)$ we obtain a {\it stabilizer quantum code with parameters} $[[73,52,7]]_8$. A  code with parameters $[73,63,7]_8$ is the best known binary linear code with length $73$ and minimum distance $7$ \cite{grassl}. In the unlikely case it were self-orthogonal for the Euclidean inner product, by using the CSS construction, we would get a $[[73,53,7]]_8$ code, which has only one unit more of dimension than ours.
}
\end{exa}

\begin{exa}
\label{ex3}
{\rm
Now we are going to give a couple of binary quantum codes obtained from the procedure described in Section \ref{qce}. Set $n=146$ and consider the following polynomials in $\mathbf{F}_2[x]$: $f(x)=1$, $h(x)= x^5 + x^4 + x^2 + x + 1$ and $g_i(x) = h(x) f_i(x)$, $1 \leq i \leq 3$, where
$$f_1(x)=x^9 + x^7 + x^4 + x^3 + 1,$$ $$f_2(x)=x^{18} + x^{16} + x^{12} + x^{10} + x^{9} + x^6 + x^4 + x^3 + x^2 + x +1 \;\;\; \mathrm{and}$$
\begin{multline*}
f_3(x)= x^{27} + x^{26} + x^{25} + x^{24} + x^{21} + x^{20} + x^{19} + x^{18} + x^{17} +
x^{16} \\+ x^{15} + x^{14} + x^{13} + x^{12} + x^{10} + x^9 + x^8 + x^6 + x^4 + x^3 + x^2
+ x + 1.
\end{multline*}
The polynomials $f_1(x)$, $f_2(x)$ and $f_3(x)$  are determined, respectively, by cyclotomic cosets $\mathcal{I}_1$, $\mathcal{I}_2$ and $\mathcal{I}_3$, with respect to $2$ modulo $n=73$, such that $\mathcal{I}_1 \subset \mathcal{I}_2 \subset \mathcal{I}_3$. They provide linear codes with parameters $[73,64,3]_2, [73,55,5]_2$ and $[73,46,9]_2$.

Consider the QC codes $Q_i := Q_2(f, g_i,h)$, $1 \leq i \leq 3$. By construction of the polynomials $g_i(x)$ and by Proposition \ref{theEuself}, it holds that
\[
Q_1^\perp \subseteq Q_2^\perp \subseteq Q_3^\perp \subseteq Q_3 \subseteq Q_2 \subseteq Q_1.
\]
Therefore, using the CSS procedure, we  get binary stabilizer quantum codes $C_1$, $C_2$ and $C_3$ with parameters $[[146,128,3]]_2$, $[[146,110,5]]_2$ and $[[146,74,8]]_2$. Now, the Steane's enlargement of the codes $C_1$ and $C_2$ \cite{Q48} gives rise to a {\it binary stabilizer quantum code with parameters} $[[146,119,5]]_2$ which exceeds the Gilbert-Varshamov bounds \cite{13Q, 17Q, matsumotouematsu00}. Steane's enlargement for $C_2$ and $C_3$ produces another stabilizer quantum code with parameters $[[146,92,8]]_2$.
}
\end{exa}

\begin{exa}
\label{ex4}
{\rm Our last example is obtained by applying Theorem \ref{hher} where Hermitian inner product is used. Write $n=80$, $q=3$ and consider the following polynomials in $\mathbf{F}_9[x]$, which involve a primitive element $\zeta$ of $\mathbf{F}_9$, $f(x)=x + \zeta^5$, $h(x)= x^2 + \zeta^7 x + \zeta$ and
$$g(x)= x^9 + 2 x^8 + \zeta^2 x^6 + 2 x^5 + \zeta^5 x^4 + \zeta x^3 + \zeta^3 x^2 + \zeta^2.$$
These polynomials satisfy the requirements of Theorem \ref{hher} and, as a consequence we get a {\it stabilizer quantum code with parameters} $[[160,140,5]]_3$.

Notice that this code exceeds the Gilbert-Varshamov bounds \cite{13Q, 17Q, matsumotouematsu00}. In addition, according to \cite{grassl},  a linear code with parameters $[160, 149,5]_3$ is the best known linear ternary code with length $n=160$ and minimum distance $d=5$. In the unlikely case, it were self-orthogonal, the CSS procedure would give a quantum code with parameters $[[160, 138, 5]]_3$, which is worse than ours. We conclude by observing that we cannot reproduce this last procedure for self-orthogonality with respect to Hermitian duality because examples of length $160$ over $\mathbf{F}_9$ are not provided in \cite{grassl}.
}
\end{exa}


\end{document}